\documentclass[aps,pra,onecolumn,nobibnotes,superscriptaddress,nofootinbib]{revtex4-2}
\usepackage{graphicx} % Required for inserting images
\usepackage{caption,graphics,physics}
\usepackage{subcaption}
\usepackage{amsmath,xcolor,bm}
\usepackage{qcircuit}
\usepackage{mathtools,amsmath,amsthm,amsfonts,amssymb,amscd}
\usepackage{hyperref}
\usepackage[ruled,vlined]{algorithm2e}
\SetAlgorithmName{Protocol}{Protocol}{List of Protocols}
\usepackage{algpseudocode}

\usepackage[capitalize]{cleveref}
\crefname{algorithm}{Protocol}{Protocols}
\hypersetup{
    colorlinks=true,
    linkcolor=red,
    filecolor=magenta,      
    urlcolor=cyan,
}
\usepackage{geometry}

\newtheorem{theorem}{Theorem}

\newtheorem{lemma}{Lemma}

\newtheorem{example}{Example}

\newtheorem{corollary}{Corollary}

\Crefname{appsec}{Appendix}{Appendices}

\begin{document}
\title{Fermionic-Adapted Shadow Tomography for dynamical correlation functions}

\author{Taehee Ko}
\email{kthmomo@kias.re.kr}
\affiliation{School of Computational Sciences, Korea Institute for Advanced Study}

\author{Mancheon Han}
\affiliation{School of Computational Sciences, Korea Institute for Advanced Study}

\author{Hyowon Park}
\affiliation{Materials Science Division, Argonne National Laboratory, Argonne, IL, 60439, USA}
\affiliation{Department of Physics, University of Illinois at Chicago, Chicago, IL, 60607, USA}

\author{Sangkook Choi}
\email{sangkookchoi@kias.re.kr}
\affiliation{School of Computational Sciences, Korea Institute for Advanced Study}

\begin{abstract}
\begin{center}
    \textbf{Abstract}
\end{center}
    Dynamical correlation functions are essential for characterizing the response of  quantum many-body systems to  external perturbation. As their calculation is classically intractable in general,  quantum algorithms are promising in this aspect, but most rely on brute force measurement strategies that evaluate one body observable pair per circuit. In this work, we introduce Fermionic-Adapted Shadow Tomography (FAST) protocols, a new framework for the efficient calculation of multiple dynamical correlation functions. The key idea is to reformulate these functions into forms that are compatible with shadow tomography techniques. The circuits in our protocols require at most two-copy measurements with uncontrolled Hamiltonian simulation. We show that the proposed protocols enhance sample efficiency and/or reduce the number of measurement circuits by one or two orders of magnitude with respect to the system size across a range of scenarios.

\end{abstract}
\maketitle

\section*{Introduction}
The simulation of quantum many-body systems is considered as one of the most promising areas for realizing quantum advantage \cite{lloyd1996universal,aspuru2005simulated,huang2020predicting,bauer2020quantum}. Among various applications, of critical interest in computational physics is the estimation of dynamical correlation functions. Since they formulate physical properties for measured outcomes from experiments \cite{mahan2013many},  developing efficient simulation tools for  the correlation functions is of practical importance for understanding the physics behind systems in experiments. Besides, correlation functions are the building-block for the functional approaches to the quantum many-body theory \cite{Baym_Kadanoff-ConservationLaws-Phys.Rev.-1961,Luttinger_Ward-GroundStateEnergy-Phys.Rev.-1960,Potthoff_Potthoff-SelfenergyfunctionalApproach-Eur.Phys.J.B-2003}.

% applied to the simulation of strongly-correlated systems such as dynamic mean-field theory approaches \cite{bauer2016hybrid}.

For the task, numerous quantum algorithms have been proposed recently. Roughly, the algorithms are classified into variational and sampling-based approaches. Variational approaches optimize parameterized circuits to approximate the correlation functions, which is suitable for near-term quantum computers  \cite{endo2020calculation,kanasugi2023computation,mootz2024adaptive,gomes2023computing,chen2021variational}. Alternatively, sampling-based approaches enable the estimation of correlation functions using guaranteed numerical techniques. These include the Hadamard test combined with the Hamiltonian simulation \cite{kokcu2024linear,bauer2016hybrid,kosugi2020linear,chiesa2019quantum,han2025quantum} or the block encoding \cite{tong2021fast}. In these algorithms,  a correlation function is estimated at a given time or frequency point.

Notably, measurements with sophisticated circuit constructions or techniques can facilitate the computation of correlation functions. For example, computing the functions at multiple time points is performed with a single circuit \cite{piccinelli2025efficient,huggins2022nearly,wang2024snapshotting}, and the methods in   \cite{king2025triply,irmejs2025approximating} apply shadow tomography techniques for computing multiple correlation functions at a time point. A key insight from these approaches is that building clever measurement strategies can lead to efficient and novel algorithms, an idea already explored in a variety of applications such as efficient measurement schemes for variational quantum algorithms \cite{izmaylov2019unitary,hamamura2020efficient,gokhale2020optimization,zhao2020measurement,zhu2024optimizing,bonet2020nearly,cerezo2021cost}, measurement-based quantum computing algorithms \cite{ferguson2021measurement,kanno2025efficient,ding2024simulating,baumer2024quantum,chowdhury2025first}, and shadow tomography techniques for learning tasks \cite{zhao2021fermionic,cho2025entanglement,huang2020predicting,hu2025demonstration,bertoni2024shallow,king2025triply}.

Particularly in condensed matter theory and quantum chemistry, estimating multiple dynamical correlation functions is an essential task, as the number of those often scales polynomially with system size. Examples include the retarded Green's function and the density-density correlation, for which the number of correlation functions scales quadratically and quartically, respectively \cite{mahan2013many}. Under a brute-force measurement strategy, these scalings translate directly into the overall sample complexity and the number of measurement circuits. Recently, several methods have been proposed to simultaneously estimate dynamical correlation functions, such as Green's function estimation using shadow tomography techniques \cite{king2025triply, irmejs2025approximating} and more general functions using a phase kickback technique \cite{huggins2022nearly}. Still, a sample-efficient and practically viable framework for simultaneously estimating many dynamical correlation functions remains lacking. 

% Despite these advances, developing a sample-efficient and practically applicable method for estimating many dynamical correlation functions.
% }

In this work, we introduce a new framework for efficiently estimating the dynamical correlation functions via shadow tomography techniques. We first reformulate the correlation functions into direct measurement forms that enable to incorporate shadow tomography techniques.  We refer to our protocols as Fermionic-Adapted Shadow Tomography (FAST). The FAST is classified into the commutator and anti-commutator cases. For the commutator correlation functions, the FAST achieves at least order one of reduction with respect to the number of fermionic modes $n$ in terms of sample complexity and the number of circuits to execute.   In the anti-commutator case, in the regime where $n\geq\frac{1}{\epsilon^2}$ for a given precision $\epsilon$, the FAST achieves improvement compared to the brute force measurement in terms of sample complexity and the number of circuits to execute by an order of magnitude.  When $n\leq \frac{1}{\epsilon^2}$, the FAST can reduce the number of circuits an order of magnitude but does not improve sample complexity.

Throughout this work, we do not consider various noises in quantum hardware, which is beyond the scope of this work. Accordingly, we make the assumption that the circuits in \cref{fig: circuits}
are implemented on a fault-tolerant quantum computer. Additionally, a problem-dependent state $\rho$ is provided, along with uncontrolled Hamiltonian simulation. The main operations in the circuits \cref{fig: circuits} involve the preparation of a problem-dependent quantum state \cite{feniou2024sparse,tang2021qubit,dallaire2016quantum,larsen2024feedback,consiglio2024variational,xie2022variational,nigmatullin2024experimental} (e.g. the ground state of system), the compilation of uncontrolled Hamiltonian simulations  \cite{jin2025partially,poulin2014trotter,ikeda2024measuring,dallaire2016quantum,zhuk2024trotter,clinton2021hamiltonian}, and the dynamical circuit strategy in \cite{kanno2025efficient}. The first two operations are among the most actively studied areas in quantum computing. The dynamical circuit strategy in \cite{kanno2025efficient} allows a single circuit to conduct the random Pauli measurement \cite{huang2020predicting}. Recent advances in quantum hardware are making the implementation of dynamic circuits progressively more practical \cite{corcoles2021exploiting,andersen2019entanglement,baumer2024efficient}. 

Under this assumption, we address a central question: can we develop more efficient protocols for estimating the dynamical correlation functions than the brute force strategy of measuring one-body observable pairs at a time, specifically in terms of sample complexity and/or the number of measurement circuits? We answer this by proposing protocols tailored to each of the representative fermion-to-qubit mappings: Jordan-Wigner (JW), Bravyi-Kitaev (BK) \cite{bravyi2002fermionic}, and ternary tree (TT) \cite{jiang2020optimal}.

\begin{figure}[htbp]

    \centering
    \begin{subfigure}[b]{0.48\textwidth}
    \centering
    \text{
\Qcircuit @C=0.9em @R=1.4em {
& \lstick{\rho}  & \gate{\exp(i\frac{\pi}{4}B)} & \gate{e^{-iHt}} & \gate{U} & \meter & \qw   }} 
\caption{Circuit for single-copy measurements without ancilla qubit}\label{fig: circuit a}
    \end{subfigure}
    \begin{subfigure}[b]{0.48\textwidth}
    \centering
    \text{
\Qcircuit @C=0.9em @R=1.4em {
& \lstick{\ket{0}^n} & \gate{G}    & \gate{\text{Initialize }\rho}  & \gate{\exp(i\frac{\pi}{4}B)} & \gate{e^{-iHt}} & \gate{F} & \meter & \qw   \\
& \cw & \cw \cwx[-1] & \cw  & \cw & \cw & \cw  \cwx[-1] & \cw }} 
\caption{Dynamic circuit for single-copy measurements without ancilla qubit}\label{fig: circuit b}
    \end{subfigure}

    \centering
    \begin{subfigure}[b]{0.48\textwidth}
    \centering
    \text{
\Qcircuit @C=0.9em @R=1.4em {
& \lstick{\ket{0}}  & \gate{\text{H}}  &\ctrl{1}  &\gate{\text{H}} \barrier[1.35em]{1} &  \meter & \qw  \\
& \lstick{\rho} & \qw & \gate{B} & \gate{e^{-iHt}} & \qw & \qw   }} 
\caption{Circuit for single-copy measurements with an ancilla qubit}\label{fig: circuit c}
    \end{subfigure}
    \begin{subfigure}[b]{0.48\textwidth}
    \centering
    \text{
\Qcircuit @C=0.9em @R=1.4em {
& \lstick{\ket{0}} & \qw  & \gate{\text{H}}  &\ctrl{1}  &\gate{\text{H}} \barrier[1.35em]{1} &  \meter & \qw  \\
& \lstick{\ket{0}^n} & \gate{G}   & \gate{\text{Initialize }\rho} & \gate{B} & \gate{e^{-iHt}} & \gate{F} & \meter & \qw   \\
& \cw & \cw \cwx[-1] & \cw & \cw  & \cw & \cw  \cwx[-1] & \cw \\}} 
\caption{Dynamic circuit for single-copy measurements with an ancilla qubit}\label{fig: circuit d}
    \end{subfigure}

\begin{subfigure}[b]{0.48\textwidth}
\text{
\Qcircuit @C=0.9em @R=1.4em {
& \lstick{\rho} & \qw & \gate{\exp(i\frac{\pi}{4}B)} & \gate{e^{-iHt}}  & \qw & \multigate{1}{\text{Bell}} & \qw & \meter  & \qw \\
& \lstick{\rho} & \qw & \gate{\exp(i\frac{\pi}{4}B)} & \gate{e^{-iHt}}  & \qw & \ghost{\text{Bell}} & \qw & \meter  & \qw }}     
\caption{Circuit for two-copy measurements without ancilla qubit}\label{fig: circuit e}
\end{subfigure}
\begin{subfigure}[b]{0.48\textwidth}
\text{
\Qcircuit @C=0.9em @R=1.4em {
& \lstick{\ket{0}}  & \gate{\text{H}}  &\ctrl{1}  &\gate{\text{H}} \barrier[1.35em]{1} &  \meter & \qw  \\
& \lstick{\rho} & \qw & \gate{B} & \gate{e^{-iHt}} & \qw & \qw & \multigate{2}{\text{Bell}} & \qw & \meter  & \qw \\
& \lstick{\ket{0}}  & \gate{\text{H}}  &\ctrl{1}  &\gate{\text{H}} \barrier[1.35em]{1} &  \meter & \qw  \\
& \lstick{\rho} & \qw & \gate{B} & \gate{e^{-iHt}} & \qw & \qw & \ghost{\text{Bell}} & \qw & \meter  & \qw }}     
\caption{Circuit for two-copy measurements with two ancilla qubits}\label{fig: circuit f}
\end{subfigure}

    \caption{Circuits for single-copy and two-copy measurements used in our protocols for estimating the dynamical correlation functions. The gate $B$ denotes a given Pauli string.   In \cref{fig: circuit a},  a set of Clifford unitaries for mutually commuting observables, denoted as $U$, is used.  In \cref{fig: circuit b}, \cref{fig: circuit d}, we use the dynamic circuit technique for random Pauli measurement \cite{kanno2025efficient}. Following the same notation in \cite{kanno2025efficient}, we denote by $G,F$, the operations for generating the probability distribution for which Pauli basis is sampled, and receiving the feedback for constructing the corresponding Pauli basis. In \cref{fig: circuit e}, we perform the Bell sampling to throw out observables of negligible expectation magnitudes \cite{huang2021information,king2025triply}. In \cref{fig: circuit f}, we also apply Bell sampling in a certain case, while in the other, we use a variant that requires measurement in a problem-dependent basis, as will be shown in \cref{lem: ours}.    }\label{fig: circuits}
\end{figure}

\section*{Results}

\subsection*{Related work}

Several works have demonstrated the estimation of Green's function using Hadamard-test circuits  \cite{endo2020calculation,kanasugi2023computation,mootz2024adaptive,chen2021variational}, based on strategies of brute-force measurements. These methods require two controlled Pauli operators in estimating real-time Green's functions, together with a time evolution block or a corresponding parameterized circuit block. Since the two Pauli operators arise in the expansion of Green's function, the number of measurement circuits and the total sample complexity scales as $\mathcal{O}(n^2)$ and $\mathcal{O}(\frac{n^2\log n}{\epsilon^2})$, respectively, where $n$ denotes the system size.  The recent works  \cite{king2025triply,irmejs2025approximating} propose methods for estimating Green's functions using shadow tomography techniques. While their target calculations are slightly different (i.e. \cite{king2025triply} for a fixed time and \cite{irmejs2025approximating} for a fixed frequency), the common idea is that  reformulating  correlation functions enables the use of shadow tomography techniques and the savings of measurement cost. The algorithm in \cite{irmejs2025approximating} is based on a continued fraction of the Green's function and outputs single and multiple Green's functions. In their approach, the number of observables scales exponentially with the level of continued fraction, so efficient measurements are considered such as classical shadow and fermionic shadow. On the other hand, the approach in \cite{king2025triply} uses the Taylor expansion of Green's function  and requires to compute derivatives at time zero. Their approach requires at least $\mathcal{O}(n^3)$ time because of the entangled measurement constructions \cite{aaronson2004improved}.  Especially for typical Hamiltonian models where the number of Majorana operators scales linearly with the system size, such as the Hubbard model, the time complexity is expected to scale as $\Omega(n^3)$ due to the processing time to estimate nested commutator terms, and the number of measurement circuits may scale as $\Omega(\frac{n}{\epsilon^2})$ since at least one-body observables have to be estimated in the process. In contrast to these methods, our approach considers both regimes: $n\leq\frac{1}{\epsilon^2}$ and $n\ge \frac{1}{\epsilon^2}$. In both cases, it can reduce the number of circuits by an order of magnitude and improves sample complexity, especially in the regime $n\ge\frac{1}{\epsilon^2}$. Additionally, our approach maintains a post-processing time comparable to that of brute-force methods, and it can estimate Green's functions across a similar range of time points using identical time-evolution blocks.  \cref{table: comparison 1} details the complexity measures across brute-force measurements \cite{endo2020calculation,kanasugi2023computation,mootz2024adaptive,chen2021variational}, the shadow method in \cite{king2025triply}, and our proposed approach.

Compared to the case of Green's function, fewer results exist for estimating bosonic correlation functions via quantum computation \cite{mootz2024adaptive,kosugi2020linear,kokcu2024linear}. With  explicit circuit constructions, \cite{kosugi2020linear} proposes an  approach for estimating response functions in the frequency domain, but the computational complexity of their approach scales exponentially with system size, due to calculations for all eigenvalues. More closely related to our work are the time-domain approaches \cite{mootz2024adaptive,kokcu2024linear}. The method in \cite{kokcu2024linear} provides an ancilla-free implementation using time-dependent Hamiltonian simulations within a linear response framework, whereas \cite{mootz2024adaptive} employs Hadamard-test circuits with multiple controlled operations to estimate general susceptibilities. Generally, these approaches incur an $\mathcal{O}(n^4\log n)$ sample complexity and a similarly scaling number of measurement circuits. By leveraging shadow tomography, our approach reduces this quartic scaling to cubic or better, depending on the regime. Moreover, similar to \cite{kokcu2024linear}, our method avoids controlled Pauli operations entirely. For all three time-domain approaches, $\mathcal{O}(n^4)$ correlation functions are evaluated from measurement outcomes, yielding an $\mathcal{O}(n^4)$ post-processing time. \cref{table: comparison 2} summarizes this comparison between our method and those in \cite{kokcu2024linear,mootz2024adaptive}.

Our technical contributions are summarized as follows:

\begin{itemize}
    \item  We derive identities \eqref{eq: bose rep}, \eqref{eq: plus rep}, and \eqref{eq: minus rep}, which establish a general framework for estimating both commutator and anti-commutator dynamical correlation functions. Our approach is naturally compatible with various shadow tomography techniques, without controlled time evolution.
    
    \item (Commutator case) By integrating measurement techniques in \cite{bonet2020nearly} and \cite{king2025triply} into our framework, we achieve a significant reduction in both sample complexity and the number of measurement circuits. For the regime $n \le \frac{1}{\epsilon^2}$ using TT mapping, we introduce a new measurement strategy (\Cref{lem: ternary}) combining classical shadows \cite{huang2020predicting} and dynamic circuits \cite{kanno2025efficient}. This approach reduces circuit construction time and measurement overhead by at least an order of magnitude relative to system size compared to previous results \cite{kokcu2024linear, mootz2024adaptive}.

    \item (Anti-commutator case) We provide tailored strategies for the anti-commutator based on the mapping and regime. In the regime $n\le\frac{1}{\epsilon^2}$, our protocol reduces only the number of circuits by an order of magnitude with respect to system size.  In the other regime $n\ge\frac{1}{\epsilon^2}$, we introduce two measurement strategies. The first is optimized for BK and TT mappings (\cref{lem: ours0}), while the second is designed for JW mapping (\cref{lem: ours}). For the JW case, we develop a novel two-copy measurement strategy that reduces the required number of circuits by a factor of $1/\epsilon^2$ relative to our first strategy.
\end{itemize}

\begin{table}[h]
\centering
\begin{tabular}{|c|c|c|c|c|}
\hline
Method & \cite{endo2020calculation,kanasugi2023computation,mootz2024adaptive,chen2021variational}   & \cite{king2025triply} & \multicolumn{2}{c|}{This work}\\ 
\hline
Measurement type & Brute-force & Shadow-based & \multicolumn{2}{c|}{Shadow-based} \\
\hline 
Target calculation & $\{G_{ab}(t)\}_{a,b=1}^n$ & $\{G_{ab}^{(q)}(0)\}_{a,b=1}^{n}$ & \multicolumn{2}{c|}{$\{G_{ab}(t)\}_{a,b=1}^n$}\\
\hline
Regime & $n\ge 1$ & $n\ge \frac{1}{\epsilon^2},\; t\approx 0$ & $n\le \frac{1}{\epsilon^2}$ & $n\ge \frac{1}{\epsilon^2}$ \\ 
\hline
Sample complexity\footnote{We consider typical Hamiltonian models such as the Hubbard model, where  $s=\mathcal{O}(n)$. This affects the sparsity parameter in \cite[Definition B.2]{king2025triply}.}  & $\mathcal{O}(\frac{n^2\log n}{\epsilon^2})$ & $\mathcal{O}(\frac{(nq)^{5q}\log n}{\epsilon^4})$ & $\mathcal{O}(\frac{n^2\log n}{\epsilon^2})$ &  $\mathcal{O}(\frac{n\log n}{\epsilon^4})$ \\
\hline 
Ancilla qubit & $1$ & $0$ & $1$  & $2$ \\
\hline
the number of copies of $\rho$ & $1$ & $2$ & $1$ & $2$ \\
\hline
the number of circuits & $\mathcal{O}(n^2)$ & $\Omega(\frac{n}{\epsilon^2})$ & $\mathcal{O}(n^2)$ (JW,BK) or $\mathcal{O}(n)$ (TT) & $\mathcal{O}(n)$ \\
\hline
Measurement circuit depth & $\mathcal{O}(1)$ & $\mathcal{O}(n)$ &  $\mathcal{O}(1)$ & $\mathcal{O}(1)$\\
\hline 
Controlled Pauli operation  & $2$    & $0$ & $1$  & $2$   \\
\hline
Time evolution block\footnote{Parameterized circuit blocks in \cite{endo2020calculation,kanasugi2023computation,mootz2024adaptive,chen2021variational} are considered as time evolution blocks} & $1$ & $0$ &  $1$ &$2$\\
\hline
Classical post-processing & $\mathcal{O}(n^2)$ & $\Omega(n^3)$ & $\mathcal{O}(n^2)$ & $\mathcal{O}(n^2)$\\
\hline
\end{tabular}
\caption{Comparison between the methods in \cite{endo2020calculation,kanasugi2023computation,mootz2024adaptive,chen2021variational}, the method in \cite[Theorem B.1]{king2025triply}, and our method, towards the calculation of the Green's function $G_{ab}(t)=-i\tr(\rho\{c_i(t),c_j^\dagger\})$. As a subroutine, the approach in \cite{king2025triply} estimates the $q$-th order derivative of Green's function.   "Controlled Pauli operation" indicates the number of controlled Pauli operations that apply a Pauli string, $P$, on the target register. "Time evolution block" corresponds to a query of the Hamiltonian simulation. "Classical post-processing" means the time spent on thresholding small measurement values and evaluating the $\mathcal{O}(n^2)$ correlation functions from measurement values.  }
\label{table: comparison 1}
\end{table}

\begin{table}[h]
\centering
\begin{tabular}{|c|c|c|c|c|}
\hline
Method & \cite{mootz2024adaptive} & \cite{kokcu2024linear} & \multicolumn{2}{c|}{This work}\\ 
\hline
Measurement type & Brute-force & Brute-force & \multicolumn{2}{c|}{Shadow-based} \\
\hline 
Target calculation &  \multicolumn{4}{c|}{$\{F_{ijkl}(t)\}_{i,j,k,l=1}^n$}\\
\hline
Regime & $n\ge 1$ & $n\ge 1$ & $n\le \frac{1}{\epsilon^2}$ & $n\ge \frac{1}{\epsilon^2}$  \\
\hline
Sample complexity  & $\mathcal{O}(\frac{n^4\log n}{\epsilon^2})$ & $\mathcal{O}(\frac{n^4\log n}{\epsilon^2})$ & $\mathcal{O}(\frac{n^3\log n}{\epsilon^2})$ & $\mathcal{O}(\frac{n^2\log n}{\epsilon^4})$ \\
\hline
Ancilla qubit & $1$ & $0$ & $0$  & $0$ \\
\hline
the number of copies of $\rho$ & $1$ & $1$ & $1$ & $2$ \\
\hline
the number of circuits & $\mathcal{O}(n^4)$ & $\mathcal{O}(n^4)$ & $\mathcal{O}(n^3)$ (JW,BK) or $\mathcal{O}(n^2)$ (TT) & $\mathcal{O}(\frac{n^2}{\epsilon^2})$ \\
\hline
Measurement circuit depth & $\mathcal{O}(1)$ & $\mathcal{O}(1)$ & $\mathcal{O}(n)$ (JW,BK) or $\mathcal{O}(1)$ (TT) &  $\mathcal{O}(n)$ \\
\hline 
Controlled Pauli operation\footnote{We let $P_4,P_5$ be the identities in \cite[Figure 2(a)]{mootz2024adaptive} } & $4$ & $0$ & $0$ & $0$   \\
\hline 
Time evolution block\footnote{We plug $t=0$ in \cite[Figure 2(a)]{mootz2024adaptive} for the comparison} & $1$ & $1$ & $1$ & $2$  \\
\hline
Classical post-processing & $\mathcal{O}(n^4)$ & $\mathcal{O}(n^4)$ & $\mathcal{O}(n^4)$ & $\mathcal{O}(n^4)$ \\
\hline
\end{tabular}
\caption{Comparison between the method in \cite{mootz2024adaptive} and our method for the calculation of the correlation function $F_{ijkl}(t)=\tr(\rho [e^{iHt}c_i^\dagger c_je^{-iHt},c_k^\dagger c_l])$. We use the same metrics as in \cref{table: comparison 1}.
 }
\label{table: comparison 2}
\end{table}

\begin{table}[t]
\centering
\renewcommand{\arraystretch}{1.25}
\begin{tabular}{ll}
\hline
\textbf{Notation} & \textbf{Description} \\
\hline

$c_i^\dagger,\, c_i$ 
& Fermionic creation and annihilation operators on mode $i\in[n]$ \\

$\{c_i,c_j^\dagger\}=\delta_{ij}$ 
& Canonical anticommutation relations \\

$\{c_i,c_j\}=0$ 
& Fermionic annihilation anticommutation relation \\

$O(t) := e^{iHt} O e^{-iHt}$ 
& Heisenberg time evolution of operator $O$ \\

$\gamma_{2i-1} = c_i + c_i^\dagger$ 
& Majorana operator (real component) \\

$\gamma_{2i} = i(c_i - c_i^\dagger)$ 
& Majorana operator (imaginary component) \\

$\{\gamma_a,\gamma_b\} = 2\delta_{ab}$ 
& Majorana anticommutation relations \\

$B$ 
& Operator represented as Pauli strings or sums of Pauli strings (context-dependent) \\

$P = \bigotimes_{k=1}^n \sigma_k^{\alpha_k}$ 
& $n$-qubit Pauli string \\

$\sigma^\alpha \in \{I,X,Y,Z\}$ 
& Single-qubit Pauli operators \\

\hline
\end{tabular}
\caption{Summary of notation used throughout the paper.}
\label{tab:notation}
\end{table}

\subsection*{Reformulations of correlation functions}

We consider two types of dynamical correlation functions:
\begin{equation}\label{eq: function rep}
    C(A,B,t) = \tr(\rho[A(t),B]) \text{ or } \tr(\rho\{A(t),B\}).
\end{equation}Here, $A$ and $B$ can be represented by Pauli strings under a fermion-to-qubit mapping, corresponding to fermionic or Majorana operators. These lattice correlation functions form a building block for the real-space correlation function:
\begin{align}
    C(\bm r, \bm r',t)=\sum_{A,B}\Phi_A(\bm r)\Phi_B(\bm r')C(A,B,t).
\end{align}Here, $\Phi_A(\bm r)$ denotes the mode wavefunction spanning in the Hilbert space for the operator $A$.

\begin{example}[Green's functions]
    The retarded Green's function is defined as
\begin{equation}
G^{R}_{ij}(t) = -i \theta(t) \tr(\rho \{ c_i(t), c_j^\dagger(0) \} ),
\end{equation}
where \( c_i(t) = e^{iHt} c_i e^{-iHt} \) is the fermionic annihilation operator in the Heisenberg picture. This lattice Green's function is used to calculate real space Green's function:
\begin{align}
    G^R(\bm r, \bm r',t)=\sum_{i,j}\varphi_i(\bm r)\varphi_j^*(\bm r')G_{ij}^R(t),
\end{align}where $\varphi_i(\bm r)$ is the wavefunction for the mode index $i$. This real space Green's function is central to the computation of spectral functions and local density of states via the Fourier transform.
\end{example}

\begin{example}[Density-Density Response]
    
Let \( n_i = c_i^\dagger c_j \) be the density matrix operators. Then, the lattice density-density correlation function is 
\begin{align}
    \chi_{ij;kl}(t) = -i\theta(t)\tr\left(\rho [n_{ij}(t),n_{kl}(0)]\right),
\end{align}which is used to compute the real space density-density response,
\begin{align}\label{eq: ddres}
    \chi(\bm r, \bm r',t)=\sum_{ij;kl}\varphi_i^*(\bm r)\varphi_j(\bm r)\varphi_k^*(\bm r')\varphi_l(\bm r')\chi_{ij;kl}(t).
\end{align}This function measures how a density perturbation at $\bm r'$ affects the density at $\bm r$ at a later time.
\end{example}

 We notice that the correlation function of the form \eqref{eq: function rep} is decomposed into 
\begin{align}\label{eq: pauli rep}
    \tr(\rho[A(t),B])=\sum_{m=1}^Mb_m \tr(\rho[A(t),B_m])\quad \text{    or  }\quad 
    \tr(\rho\{A(t),B\})=\sum_{m=1}^Mb_m \tr(\rho\{A(t),B_m\}),
\end{align}if $B=\sum_{m=1}^Mb_m B_m$. For our applications, $B_m$ denotes a Pauli string. When $B$ is a $1$-body fermionic observable or a creation operator (e.g. the examples below \eqref{eq: function rep}), $M=\mathcal{O}(1)$ and $b_m=\mathcal{O}(1)$ for each $m$, under a fermion-to-qubit mapping \cite{endo2020calculation,kanasugi2023computation}.  

From this observation, 
we present reformulations of the correlation functions in~\eqref{eq: function rep}, focusing on the case where $B$ is a Pauli string. The general case where $B=\sum_{m=1}^Mb_m B_m$ can be straightforwardly resolved by repeating those reformulations $M$ times and summing them up.

In the commutator case, we reformulate the representation of the correlation functions as,
\begin{equation}
    \begin{split}
        C(A,B,t)&=\tr(\rho[A(t),B]) \\
        & = -i\left[4\tr(\left(\frac{I+iB^\dagger }{2}\right)\rho\left(\frac{I-iB}{2}\right)A(t))-\tr(\rho A(t))-\tr(B^\dagger\rho BA(t))\right].
    \end{split}
\end{equation}If $B$ is a Pauli string, we can simplify this formulation to
\begin{equation}\label{eq: bose rep}
\begin{split}
    &C(A,B,t) = -i\left[2\tr(e^{-iHt}e^{i\frac{\pi}{4}B}\rho e^{-i\frac{\pi}{4}B}e^{iHt}A)-\tr(e^{-iHt}\rho e^{iHt}A)-\tr(e^{-iHt}B\rho Be^{iHt}A)\right]
\end{split}
\end{equation}
Therefore, this representation allows us to estimate the commutator correlation function with three types of direct measurements and three circuits: 1. measure observable $A$ with state 
$e^{-iHt}e^{i\frac{\pi}{4}B}\rho e^{-i\frac{\pi}{4}B}e^{iHt}$, 2. measure observable $A$ with state $e^{-iHt}\rho e^{iHt}$, and 3. measure observable $A$ with state $e^{-iHt}B\rho Be^{iHt}$.

However, the anti-commutator case is not straightforward as the commutator case. To do similarly, we require the operation $\frac{I+B}{2}$ or $\frac{I-B}{2}$ as shown below but these are non-unitary operations that are not directly implementable on quantum computers. In this case, we therefore introduce a measurement-based strategy, which leads to two representations of the anti-commutator correlation function. By a bit of algebra, we achieve that
\begin{equation}
    \begin{split}
        C(A,B,t)&=\tr(\rho\{A(t),B\}) \\
        & = 4\tr(\left(\frac{I+B}{2}\right)\rho\left(\frac{I+B}{2}\right)A(t))-\tr(\rho A(t))-\tr(B\rho B A(t)),
    \end{split}
\end{equation}and
\begin{equation}
    \begin{split}
        C(A,B,t)&=\tr(\rho\{A(t),B\}) \\
        & = -4\tr(\left(\frac{I-B}{2}\right)\rho\left(\frac{I-B}{2}\right)A(t))+\tr(\rho A(t))+\tr(B\rho B A(t)).
    \end{split}
\end{equation}Equivalently, we observe that
\begin{equation}\label{eq: plus rep}
    \begin{split}
        C(A,B,t)& = 4C_{+}^2\tr(\rho_{+}A)-\tr(e^{-iHt}\rho e^{iHt}A)-\tr(e^{-iHt}B\rho Be^{iHt}A)\\
        &C_{+}^2=\tr(\left(\frac{I+B}{2}\right)\rho \left(\frac{I+B}{2}\right)),\; \rho_{+}= \frac{e^{-iHt}\left(\frac{I+B}{2}\right)\rho \left(\frac{I+B}{2}\right) e^{iHt}}{C_{+}},
    \end{split}
\end{equation}and
\begin{equation}\label{eq: minus rep}
    \begin{split}
        C(A,B,t)& = -4C_{-}^2\tr(\rho_{-}A)+\tr(e^{-iHt}\rho e^{iHt}A)+\tr(e^{-iHt}B\rho Be^{iHt}A)\\
        & C_{-}^2=\tr(\left(\frac{I-B}{2}\right)\rho \left(\frac{I-B}{2}\right)),\; \rho_{-}= \frac{e^{-iHt}\left(\frac{I-B}{2}\right)\rho \left(\frac{I-B}{2}\right) e^{iHt}}{C_{-}}.
    \end{split}
\end{equation}Note that $C_{+}^2+C_{-}^2=1$. These quantities are exactly the probabilities of the output state to be $\rho_{\pm }$, which are obtainable from the circuits in \cref{fig: circuit c}, \cref{fig: circuit d}, and \cref{fig: circuit f} with a controlled Pauli gate for $B$. Technically, we check this in \eqref{eq: fig1c}.

We notice that the output states $\rho_{\pm}$ from \cref{fig: circuit c}, \cref{fig: circuit d}, or \cref{fig: circuit f} are prepared probabilistically. This implies that we cannot deterministically use the representations in \eqref{eq: plus rep} and \eqref{eq: minus rep} as the commutator case. Instead, we suggest to follow a majority rule for selecting  one of the representations based on resulting statistics. Simply, for a scheduled number of circuit samplings, we count the outcomes in the ancilla. If the state $\ket{0}$ is observed more frequently than $\ket{1}$, we know that $C_{+}^2\geq C_{-}^2$ with high probability and choose the \emph{plus} representation \eqref{eq: plus rep}. In the opposite case, we choose the \emph{minus} representation \eqref{eq: minus rep}. More importantly, we can make use of the output states at least half of the scheduled number of samplings. We state this formally in the following lemma.
\begin{lemma}\label{lem: majority rule}
    Given scheduled $N_s=\mathcal{O}(\frac{\log\frac{1}{\delta}}{\epsilon^2})$ circuit samplings, let $n_{+}$ and $n_{-}$ be the number of outcomes being $\ket{0}$ and $\ket{1}$ in the ancilla qubit, respectively. Then, the following holds with probability at least $1-\delta$,
    \begin{equation}
        \abs{C_{+}^2-\frac{n_{+}}{N_s}}, \abs{C_{-}^2-\frac{n_{-}}{N_s}}\leq \epsilon,
    \end{equation}and either $\rho_{+}$ or $\rho_{-}$ is prepared with at least $\frac{N_s}{2}$ times. Consequently, in either of the fermionic correlation functions in \eqref{eq: plus rep} and \eqref{eq: minus rep}, the first term can be estimated with at least $\frac{N_s}{2}$ samples. 
\end{lemma}
\begin{proof}
    The first statement immediately follows from the Hoeffding's inequality by considering the random variable $X$ that counts the number of outcomes to be $\ket{0}$. The second statement is true since $\frac{n_{+}}{N_s}+ \frac{n_{-}}{N_s}=1$ and either of $n_{+}$ or $n_{-}$ must be $\geq \frac{N_s}{2}$, which means that one of the states $\rho_{\pm}$ is more often returned  with at least $\frac{N_s}{2}$ samples.
\end{proof}

\subsection*{Protocol for the commutator case}

We consider the commutator correlation functions that involve $1$-body observables with $i,j,k,l\in[n]$, 
\begin{equation}\label{eq: comm-correlation}
    \tr(\rho [e^{iHt}c_i^\dagger c_je^{-iHt},c_k^\dagger c_l]).
\end{equation} In this section, 
we present a protocol that allows to efficiently estimate these correlation functions. The key idea is to incorporate shadow tomography techniques into the representation \eqref{eq: bose rep} and design  protocols for the dynamical correlation functions \eqref{eq: comm-correlation}.

For a given precision $\epsilon$, we distinguish between two regimes:  $n\leq \frac{1}{\epsilon^{2}}$ and $n\geq \frac{1}{\epsilon^{2}}$. This is motivated by the recent result by King \emph{et al} \cite{king2025triply}. We first consider the case when $n\leq \frac{1}{\epsilon^2}$ and JW or BK mapping is applied. The following lemma is an immediate extension of fermionic partition strategy in \cite{bonet2020nearly} as we prove in the following. 
\begin{lemma}\label{lem: bonet}Assume $n\leq \frac{1}{\epsilon^2}$. Then,
there exists a shadow tomography protocol that is compatible with the Jordan-Wigner (JW) mapping and the Bravyi-Kitaev (BK) mapping, and uses single-copy measurements to estimate the $1$-body fermionic observables to additive error $\epsilon$ with sample complexity,
\begin{equation}
    \mathcal{O}(\frac{n\log n}{\epsilon^2}),
\end{equation}and $\mathcal{O}(n)$ measurement circuits of $\mathcal{O}(n)$ depth.
\end{lemma}
\begin{proof}
    From a fermionic partition strategy \cite{bonet2020nearly}, we can construct $\mathcal{O}(n)$ clifford measurement circuits for measuring products of two Majorana operators, with $\mathcal{O}(n)$ depth. Since 
    \begin{align}
        c_i^\dag c_j = (\gamma_{2i-1}+i\gamma_{2i})(\gamma_{2j-1}-i\gamma_{2j}),
    \end{align}we notice that the $1$-body fermionic observable is formed by a sum of four $2$-Majorana operators, and therefore measuring the $1$-body fermionic observables requires $\mathcal{O}(n)$ circuits with $\mathcal{O}(n)$ depth. 
     
\end{proof}
Now we consider the situation where $n\le\frac{1}{\epsilon^2}$ and TT mapping \cite{jiang2020optimal} is applied. Compared to the BK mapping that achieves $\log_2 n+1$ locality \cite{havlivcek2017operator} as the worst-case, the TT mapping improves the base factor as $\log_3 (2n)$ locality. In turn, this enables estimating the $1$-body observables with comparable sample complexity to the protocol in \cref{lem: bonet} by incorporating it into the random Pauli measurement in \cite{huang2020predicting}. We can further make it practical for real quantum execution by employing the notion of dynamic circuits \cite{kanno2025efficient}, which allows to do the task with a single circuit in principle. We summarize this protocol in the following lemma. 
\begin{lemma}\label{lem: ternary}
    Assume that $n\leq \frac{1}{\epsilon^2}$. Then, there exists a shadow tomography protocol that is compatible with the ternary tree (TT) mapping, and uses single-copy measurements to estimate the $1$-body observables to additive error $\epsilon$ with sample complexity,
\begin{equation}
    \mathcal{O}(\frac{n\log n}{\epsilon^2}),
\end{equation}and $\mathcal{O}(1)$ measurement circuits of $\mathcal{O}(1)$ depth.
\end{lemma}
\begin{proof}
    Since the TT mapping yields $\log_3(2n)$-local Pauli strings for the Majorana operators, and the sample complexity of the random Pauli measurement \cite{huang2020predicting} scales as $\mathcal{O}(\frac{3^w\log m}{\epsilon^2})$ where $w$ is the locality of Pauli string and $m$ is the number of observables, the sample complexity of estimating the $1$-body observables is $\mathcal{O}(\frac{n\log n}{\epsilon^2})$, as we deduced in the proof of \cref{lem: bonet}. Furthermore, the dynamic circuit \cite{kanno2025efficient} allows to perform the random Pauli measurements with $\mathcal{O}(1)$ measurement circuits.
\end{proof}

For the case where $n\geq \frac{1}{\epsilon^2}$, we use the results from \cite{king2025triply} as restated in the following lemma.
\begin{lemma}{\cite[from Lemmas 1.1 and 4.1]{king2025triply}}\label{lem: king} Assume that $n\geq \frac{1}{\epsilon^2}$. Then, 
there exists a shadow tomography protocol that is compatible with the JW, BK, and TT mappings, uses $\mathcal{O}(\frac{1}{\epsilon^2})$ circuits, single-copy and two-copy measurements to estimate the $1$-body observables to additive error $\epsilon$ with sample complexity,
\begin{equation}
    \mathcal{O}(\frac{\log n}{\epsilon^4}),
\end{equation}and the depth of each circuit scales as $\mathcal{O}(n)$ at most.
\end{lemma}For concreteness, the proof of this lemma is shown below \cref{lem: coloring alg}, which immediately follows from the original work \cite{king2025triply}.

By applying these results, we obtain a shadow tomography protocol for estimating the correlation functions in \eqref{eq: comm-correlation} as follows,
\begin{theorem}\label{thm1}
    For a given precision $\epsilon>0$, there exists a shadow tomography protocol that is compatible with the JW, BK, and TT mappings, uses at most two-copy measurements, and estimates the correlation functions in \eqref{eq: comm-correlation} to the precision $\epsilon$ with the following complexity:
    \begin{itemize}
        \item if $n\le \frac{1}{\epsilon^2}$, the sample complexity scales as  $\mathcal{O}(\frac{n^{3}\log n}{\epsilon^2})$ with $\mathcal{O}(n^3)$ (JW, BK) and $\mathcal{O}(n^2)$ (TT) measurement circuits.
        \item if $n\ge \frac{1}{\epsilon^2}$, the sample complexity scales as $\mathcal{O}(\frac{n^2\log n}{\epsilon^4})$ with $\mathcal{O}(\frac{n^2}{\epsilon^2})$ measurement circuits.
    \end{itemize}
\end{theorem}
\begin{proof}
In the representation \eqref{eq: comm-correlation}, we fix $k,l$. Then, for any given fermion-to-qubit mapping among JW, BK, and TT, we can transform the observable $c_k^\dag c_l$ into a sum of four Pauli strings. This yields an expression similar to the commutator case in \eqref{eq: pauli rep}. Now it reduces to estimating the following quantity from \eqref{eq: comm-correlation}
\begin{align}\label{eq: target 1}
    \tr(\rho [e^{iHt}c_i^\dagger c_je^{-iHt},B])
\end{align}for a Pauli string $B$ from $c_k^\dag c_l$. Using the representation \eqref{eq: bose rep} and setting $A=c_{i}^\dagger c_j$, we apply the results in \cref{lem: bonet} (JW,BK), \cref{lem: ternary} (TT), and \cref{lem: king} (JW,BK,TT) to estimate $A$ with $i,j\in[n]$, by considering the following three measurements 
\begin{equation}\label{eq: A,B bose}
    \begin{split}
        &\tr(\rho_1 A),\quad \rho_1=e^{-iHt}e^{i\frac{\pi}{4}B}\rho  e^{-i\frac{\pi}{4}B} e^{iHt} \\
        &\tr(\rho_2 A),\quad \rho_2=e^{-iHt}\rho e^{iHt}\\
        &\tr(\rho_3 A),\quad  \rho_3=e^{-iHt}B\rho B e^{iHt}.
    \end{split}
\end{equation}
Therefore, in estimating \eqref{eq: target 1},  when $n\leq\frac{1}{\epsilon^{2}}$, $\mathcal{O}(\frac{n\log n}{\epsilon^2})$ samples suffice, and $\mathcal{O}(\frac{\log n}{\epsilon^4})$ when $n\ge\frac{1}{\epsilon^2}$.  Since there are $n^2$ choices for $k,l$, the total sample complexity for estimating the correlation function \eqref{eq: comm-correlation} scales as $\mathcal{O}(\frac{n^{3}\log n}{\epsilon^2})$ when $n\leq \frac{1}{\epsilon^{2}}$ and $\mathcal{O}(\frac{n^2\log n}{\epsilon^4})$ when $n\geq \frac{1}{\epsilon^{2}}$, respectively. 
\end{proof}

\begin{algorithm}
\SetAlgoLined
	\KwData{Given initial state $\rho$, precision $\epsilon>0$}
	\KwResult{approximate the correlation functions \eqref{eq: comm-correlation} for any $i,j,k,l$}

     If $n\leq \frac{1}{\epsilon^{2}}$:
        \For{$k,l=1:n$}{

        Apply the protocol in \cref{lem: bonet} (if JW or BK mapping is applied) or in \cref{lem: ternary} (if TT is applied) for computing the quantities in \eqref{eq: A,B bose}

        Use the estimates of quantities to compute the correlation functions \eqref{eq: bose rep}
        }
     else:
    
	\For{$k,l=1:n$}{

        Apply  \cref{lem: king} for computing the quantities in \eqref{eq: A,B bose}

        Use the estimates of quantities to compute the correlation functions \eqref{eq: bose rep}
        }
	\caption{Fermionic-adapted shadow tomography 1 (FAST 1)}
    \label{alg: algorithm1}
\end{algorithm}

\subsection*{Protocol for the anti-commutator case}

We turn to the anti-commutator case. In this case, the important correlation function is the retarded Green's function (RGF),
\begin{equation}\label{eq: RGF}
    G_{ab}^R(t):=-i\Theta(t)\tr(\rho\{c_i(t),c_j^\dagger\}), 
\end{equation}where $c_i(t):=e^{iHt}c_ie^{-iHt}$ and $\rho$ is the ground state of a quantum system. For a fixed time point $t$, we present a protocol that estimates the RGF's $\{G_{ab}^R(t)\}_{a,b=1}^{n}$ to additive error $\epsilon$. Similar to \cref{alg: algorithm1}, we consider the two regimes: $n\leq \frac{1}{\epsilon^{2}}$ and $n\geq \frac{1}{\epsilon^{2}}$.

We first consider the regime where $n\leq \frac{1}{\epsilon^2}$. In this regime, we can estimate the Green's function without controlled time evolution. When JW or BK mapping is applied, we use the majority rule in \cref{lem: majority rule} with the identities \eqref{eq: plus rep} and \eqref{eq: minus rep}, with the circuit in \cref{fig: circuit c}. Then, we estimate \eqref{eq: RGF} for each pair $(i,j)$. Thus, the sample complexity and the number of measurement circuits remain comparable to brute-force measurements, scaling as $\mathcal{O}(n^2)$. If the TT mapping is used, the approach in \cref{lem: ternary} for the commutator case can be reused with the identities \eqref{eq: plus rep} and \eqref{eq: minus rep} based on \cref{lem: majority rule}. This approach reduces the number of measurement circuits by an order of magnitude with respect to system size. In \cref{alg: algorithm2}, we outline this procedure as a conditional result.

In what follows, we propose two sample-efficient measurement strategies in the regime where $n\geq \frac{1}{\epsilon^{2}}$. When BK is used and $n\ge\frac{1}{\epsilon^2}$, we derive a  measurement strategy from \cite[Lemma 2.1]{king2025triply} as follows
\begin{lemma}\label{lem: ours0}
Assume that $n\geq \frac{1}{\epsilon^2}$ and BK or TT mapping is applied. Then, there exists a shadow tomography protocol that uses $\mathcal{O}(\frac{\log (\frac{1}{\epsilon})}{\epsilon^4})$ single-copy and $\mathcal{O}( \frac{\log (n)}{\epsilon^4})$ two-copy measurements, and estimates the Majorana operators to additive error $\epsilon$, with $\mathcal{O}(\frac{\log n}{\epsilon^4})$ sample complexity and $\mathcal{O}(\frac{1}{\epsilon^2})$ circuits  of $\mathcal{O}(1)$ measurement circuit depth.    
\end{lemma}
\begin{proof}
Note that the $2n$ Majorana operators mutually anti-commute. Thus, according to \cite[Lemma 2.1]{king2025triply}, there are at most $\mathcal{O}(\frac{1}{\epsilon^2})$ observables whose expectation values are larger than $\epsilon$. Thus, we first apply the Bell sampling to threshold  observables of small magnitudes and perform brute force measurements on the remaining observables. The first step requires two-copy measurements with $\mathcal{O}(\frac{\log n}{\epsilon^4})$ sample complexity and the other requires $\mathcal{O}(\frac{1}{\epsilon^2}\cdot \frac{\log (\frac{1}{\epsilon^2})}{\epsilon^2})$ single-copy brute force measurements. Since $n\ge \frac{1}{\epsilon^2}$, the sample complexity is $\mathcal{O}(\frac{\log n}{\epsilon^4})$. Note that due to the brute force measurement, this procedure requires $\mathcal{O}(\frac{1}{\epsilon^2})$ circuits.    
\end{proof}
We remark that the procedure shown in \cref{lem: ours0} can be directly applied for the case when JW mapping is applied. However, in the following section, we propose a measurement strategy specifically tailored to the JW mapping such that the number of circuits can be reduced from $\mathcal{O}(\frac{1}{\epsilon^2})$ in \cref{lem: ours0} to $\mathcal{O}(1)$ and only $\mathcal{O}(1)$ single-copy measurements suffice, which is a significant improvement.

When the JW mapping is applied and $n\ge\frac{1}{\epsilon^2}$, we can design a more efficient measurement strategy in terms of the number of circuits and , compared to the strategy in \cref{lem: ours0}. In essence, our approach is divided into three steps as follows,
\begin{enumerate}
    \item Estimate the magnitudes of Majorana operators under the JW mapping using the Bell sampling \cite{huang2021information}

    \item Throw away negligible observables and consider the remaining observables

    \item Use a chained measurement strategy to determine the signs of the remaining observables.
\end{enumerate}The first two steps are typical in shadow tomography techniques based on Bell sampling \cite{huang2021information,king2025triply}.
In the first step, we consider the following Pauli observables,  
\begin{equation}
 \{X_1, Z_1X_2,...,Z_1\cdots Z_{n-1}X_n, Y_1, Z_1Y_2,...,Z_1\cdots Z_{n-1}Y_n\},   
\end{equation}which correspond to the Majorana operators under the JW mapping. Notice that these observables mutually anticommute. Denoting this set of observables by $\{P_i\}_{i=1}^{2n}$, we apply the known Bell sampling procedure \cite{king2025triply,huang2021information} to the set $\{P_i\otimes P_i\}_{i=1}^{2n}$, and in the second step, neglect the observables that are estimated to have magnitudes less than $\frac{3\epsilon}{4}$. 

For the last step, we introduce a variant of the Bell sampling strategy with more technical detail below \cref{lem: bell-basis}. Here we provide a brief intuition behind the strategy. In the known Bell sampling strategy, we construct the mutually-commuting observables $\{P_i\otimes P_i\}_{i=1}^{2n}$. Instead of this, we consider the following construction of observables,
\begin{equation}
    P_1\otimes P_2, P_2\otimes P_3,..., P_{2n-1}\otimes P_{2n}.
\end{equation}Notice that these observables mutually commute, since $P_j$'s mutually anticommute. Thus, we can estimate these observables simultaneously. With this in mind, if we estimate the sign of $P_1$, then we can infer that of $P_2$ from the estimation of $P_1\otimes P_2$, that of $P_3$ from $P_2\otimes P_3$, and so on. We call this a chained measurement strategy. Below \cref{lem: bell-basis}, we discuss more detail for how to perform this chained measurement strategy. We here present a result for the proposed strategy as follows.
\begin{lemma}\label{lem: ours}
    Assume that $n\geq \frac{1}{\epsilon^2}$ and the JW mapping is applied. Then, there exists a shadow tomography protocol that uses $\mathcal{O}(1)$ single-copy and $\mathcal{O}(\frac{\log n}{\epsilon^4})$ two-copy measurements, and estimates the Majorana operators to additive error $\epsilon$, with $\mathcal{O}(\frac{\log n}{\epsilon^4})$ sample complexity and $\mathcal{O}(1)$ circuits of $\mathcal{O}(1)$ measurement circuit depth. 
\end{lemma}
The proof of this lemma is shown below the proof of \cref{lem: bell-basis}. To put together, we obtain the following theorem that states the overall complexity of a protocol for estimating the RGF \eqref{eq: RGF}. 
\begin{theorem}\label{thm2}
    There exists a shadow tomography protocol that uses at most two-copy measurements, and estimates the RGF \eqref{eq: RGF} to additive error $\epsilon$ with the following regimes:
    \begin{itemize}
        \item if $n\le\frac{1}{\epsilon^2}$, no sample complexity improvement over the brute force measurement but only the number of circuits is improved as $\mathcal{O}(n)$ under TT mapping
        \item if $n\ge \frac{1}{\epsilon^2}$, the sample complexity scales  as $\mathcal{O}(\frac{n\log n}{\epsilon^4})$ under JW, BK, and TT mapping. In particular, $\mathcal{O}(\frac{n}{\epsilon^2})$ circuits are required for the cases of BK and TT, and $\mathcal{O}(n)$ circuits for JW.
    \end{itemize}
\end{theorem}
\begin{proof}
 Similar to the proof of \cref{thm1}, we fix $j$ in \eqref{eq: RGF} and consider the form of a sum of Pauli strings like \eqref{eq: pauli rep}. Then, estimating \eqref{eq: RGF} reduces to computing 
 \begin{align}
     \tr\left(\rho\{e^{iHt}c_ie^{-iHt},B\}\right)
 \end{align}for a Pauli string $B$ from $c_j^\dag$. By \cref{lem: majority rule}, we use either of the identities \eqref{eq: plus rep} and \eqref{eq: minus rep} depending on the measurement outcomes, and then apply \cref{lem: ternary} for the case where $n\le\frac{1}{\epsilon^2}$ and TT is applied, \cref{lem: ours0} for the case where $n\ge\frac{1}{\epsilon^2}$ and BK or TT is applied, and \cref{lem: ours} for the case where $n\ge \frac{1}{\epsilon^2}$ and JW is applied. Since there are $n$ choices for the index $j$, the total sample complexity and the number of measurement circuits required for each of the three cases can be derived as in the statement. 
\end{proof}

\begin{algorithm}
\SetAlgoLined
\KwData{Given initial state $\rho$, precision $\epsilon>0$}
	\KwResult{approximate the RGF \eqref{eq: RGF} for any $a,b$ (\textbf{with sample-efficiency gains only in the regime $n\ge\frac{1}{\epsilon^2}$})}

     If $n\leq \frac{1}{\epsilon^{2}}$:
        
        \For{$b=1:n$}{
        Perform the brute force measurement strategy if JW or BK mapping is applied, with the circuit in \cref{fig: circuit c}.
        
        Apply the method in \cref{lem: ternary} when TT is applied, with the circuit in \cref{fig: circuit d}.

        Use the resulting measurement outcomes to compute the RGFs by selecting either of the representations \eqref{eq: plus rep} or \eqref{eq: minus rep}.
        }
     else:
    
	\For{$b=1:n$}{

        Apply  \cref{lem: ours} (if JW mapping is applied) or \cref{lem: ours0} (if BK or TT is used)

        Use the resulting measurement outcomes to compute the RGFs by selecting either of the representations \eqref{eq: plus rep} or \eqref{eq: minus rep}.}	
        \caption{Fermionic-adapted shadow tomography 2 (FAST 2: conditional anti-commutator protocol)}
    \label{alg: algorithm2}
\end{algorithm}
Again we remark that unlike the commutator case in \cref{alg: algorithm1}, there needs an additional step of selecting either of the representations \eqref{eq: plus rep} or \eqref{eq: minus rep} in \cref{alg: algorithm2}. However, this additional step is performed without any further quantum simulations, by doing classical post-processing on the measurement results. Specifically, as discussed above \cref{lem: majority rule}, we count the number of $\ket{0}$ and $\ket{1}$ states in the ancilla qubit and know whether $\rho_{+}$ in \eqref{eq: plus rep} appears more than $\rho_{-}$ in \eqref{eq: minus rep} or vice versa. By this majority rule, we select  either of the representations \eqref{eq: plus rep} or \eqref{eq: minus rep}.

\subsection*{Complexity tables for the proposed  protocols and brute force measurements}
 
We summarize the overall complexity for \cref{alg: algorithm1} and \cref{alg: algorithm2} in  \cref{table1} and \cref{table2}, respectively. On the other hand, \cref{table3} shows the overall complexity for brute force measurements.

\begin{table}[h]
\centering
\begin{tabular}{|c|c|c|c|c|}
\hline
Correlation functions & \multicolumn{2}{|c|}{Commutator} & \multicolumn{2}{c|}{Anti-commutator} \\
\hline
\hline
Regime & $n\leq\frac{1}{\epsilon^{2}}$ & $n\geq\frac{1}{\epsilon^{2}}$ & $n\leq\frac{1}{\epsilon^{2}}$ & $n\geq\frac{1}{\epsilon^{2}}$ \\
\hline
Sample complexity & $\mathcal{O}(\frac{n^3\log n}{\epsilon^2})$ & $\mathcal{O}(\frac{n^2\log n}{\epsilon^4})$ & $\mathcal{O}(\frac{n^2\log n}{\epsilon^2})$\footnote{No improvement over the brute force strategy of measurement} & $\mathcal{O}(\frac{n\log n}{\epsilon^4})$ \\
\hline
Qubit count\footnote{For the case of TT mapping, the number of qubits equal to the system size when $n=\frac{3^h-1}{2}$ for tree height $h\ge 1$, otherwise it is generally larger but at most $2n$ \cite{jiang2020optimal}.} & $n$ & $2n$ & $n+1$ & $2n+2$\\
\hline 
Ancilla qubit & $0$ & $0$  & $1$ & $2$ \\
\hline
Copy count & $1$ & $2$ & $1$ & $2$\\
\hline
\end{tabular}
\caption{Summary of \cref{alg: algorithm1} and \cref{alg: algorithm2} for estimating the correlation functions \eqref{eq: comm-correlation} and \eqref{eq: RGF} in terms of sample complexity,  the number of qubits required (Qubit count), the number of ancilla qubits (Ancilla qubit), and the maximal number of copies of the initial quantum state used (Copy count).} 
\label{table1}
\end{table}

\begin{table}[htbp] % Changed [h] to [htbp] for better float placement
\centering
\resizebox{\textwidth}{!}{%
\begin{tabular}{|c|c|c|c|c|c|c|c|}
\hline
Correlation functions & \multicolumn{3}{|c|}{Commutator} & \multicolumn{4}{c|}{Anti-commutator} \\
\hline
\hline
Regime & $n\leq\frac{1}{\epsilon^{2}}$ & $n\leq\frac{1}{\epsilon^{2}}$ & $n\geq\frac{1}{\epsilon^{2}}$ & $n\leq\frac{1}{\epsilon^{2}}$ &$n\leq\frac{1}{\epsilon^{2}}$ & $n\geq\frac{1}{\epsilon^{2}}$ & $n\geq\frac{1}{\epsilon^{2}}$ \\
\hline
Mapping & JW,BK & TT & JW,BK,TT & JW, BK & TT & BK, TT & JW \\ 
\hline
Measurement & MMC & DC & B + MMC & BM & DC & B + BM  & B + BM\\
\hline 
Circuit construction time & $\mathcal{O}(n^6)$ & $\mathcal{O}(n^2)$ & $\mathcal{O}(\frac{n^5}{\epsilon^2})$ & $\mathcal{O}(n^2)$ & $\mathcal{O}(n)$ & $\mathcal{O}(\frac{n}{\epsilon^2})$ & $\mathcal{O}(n)$ \\
\hline
Circuit depth & $\mathcal{O}(n)$ & $\mathcal{O}(1)$ & $\mathcal{O}(n)$  & $\mathcal{O}(1)$ & $\mathcal{O}(1)$ & $\mathcal{O}(1)$  & $\mathcal{O}(1)$ \\
\hline
Circuit count & $\mathcal{O}(n^3)$ & $\mathcal{O}(n^2)$ & $\mathcal{O}(\frac{n^2}{\epsilon^2})$  & $\mathcal{O}(n^2)$ & $\mathcal{O}(n)$ & $\mathcal{O}(\frac{n}{\epsilon^2})$ & $\mathcal{O}(n)$ \\
\hline
Circuit & \cref{fig: circuit a} & \cref{fig: circuit b} & \cref{fig: circuit a},\cref{fig: circuit e} & \cref{fig: circuit c}  & \cref{fig: circuit d} & \cref{fig: circuit c},\cref{fig: circuit f}  & \cref{fig: circuit c},\cref{fig: circuit f}\\
\hline
Result & \cref{lem: bonet} & \cref{lem: ternary} & \cref{lem: king} & - &\cref{lem: ternary} & \cref{lem: ours0}  & \cref{lem: ours} \\
\hline
\end{tabular}%
}
\caption{Summary of \cref{alg: algorithm1} and \cref{alg: algorithm2} for estimating the correlation functions \eqref{eq: comm-correlation} and \eqref{eq: RGF} in terms of the fermion-to-qubit mapping (Mapping), the types of measurements (Measurement), classical time for constructing measurement circuits (Circuit construction time), the circuit depth required for the corresponding measurements (Circuit depth), the number of circuits to execute (Circuit count), circuit structure used for execution (Circuit), and the corresponding lemma (Result). The measurement types are abbreviated as follows:  Clifford measurements for mutually commuting observables (MMC), random Pauli measurements combined with dynamic circuits (DC), Bell sampling (B), and brute force measurements in the computational basis (BM). As discussed in \cite{king2025triply}, techniques in the stabilizer formalism \cite{aaronson2004improved} can be used to construct a Clifford circuit for MMC within  $\mathcal{O}(n^3)$ time. Therefore, in the cases where MMC is considered, the time complexity is estimated as $\mathcal{O}(n^3)\times \text{Circuit count}$, since MMC is performed independently for each circuit in the protocols.}
\label{table2}
\end{table}

\begin{table}[h]
\centering
\begin{tabular}{|c|c|c|}
\hline
Correlation functions & Commutator & Anti-commutator \\
\hline
\hline
Sample complexity & $\mathcal{O}(\frac{n^4\log n}{\epsilon^2})$  & $\mathcal{O}(\frac{n^2\log n}{\epsilon^2})$ \\
\hline
Qubit count\footnote{For the case of TT mapping, the number of qubits equal to the system size when $n=\frac{3^h-1}{2}$ for tree height $h\ge 1$, otherwise it is generally larger but at most $2n$ \cite{jiang2020optimal}.} & $n$ & $n+1$ \\
\hline 
Ancilla qubit & $0$ & $1$ \\
\hline
Time complexity & $\mathcal{O}(1)$ & $\mathcal{O}(1)$\\
\hline 
Circuit depth & $\mathcal{O}(1)$ & $\mathcal{O}(1)$  \\
\hline 
Circuit count & $\mathcal{O}(n^4)$ & $\mathcal{O}(n^2)$ \\
\hline
Circuit & \cref{fig: circuit a} & \cref{fig: circuit c} \\
\hline
\end{tabular}
\caption{Summary of brute force measurements for computing the correlation functions \eqref{eq: comm-correlation} and \eqref{eq: RGF}. These approaches require only single-copy measurements. }
\label{table3}
\end{table}

\subsection{Numerical result}
We validate \cref{alg: algorithm1} for the estimation of bosonic correlation functions \eqref{eq: comm-correlation} with numerical simulations. We focus on the regime $n\le\frac{1}{\epsilon^2}$, which is of practical interest as will be shown in \cref{col: boson}. We compare our protocol in \cref{lem: ternary} with the typical brute force strategy that estimates the values of the form $\tr\left(\rho [e^{iHt}\gamma_i\gamma_je^{-iHt},\gamma_k\gamma_l]\right)$ via a Hadamard test. That is, controlled Pauli operations for the Pauli strings corresponding to the products of Majorana operators, $\gamma_i\gamma_j$ and $\gamma_k\gamma_l$, are implemented, as the setting in \cref{table: comparison 2} with the circuits \cite{mootz2024adaptive}. The resulting correlation values  can be used to estimate the correlation functions in \eqref{eq: comm-correlation}.

We use the TT mapping for the two approaches. For demonstration purposes, we study a one-dimensional toy model with periodic boundary condition (i.e. $c_{n+1}=c_1$), known as the Su–Schrieffer–Heeger (SSH) model, 
\begin{align}
H = -V_{nn}\sum_{i=1}^n(1+(-1)^i\frac{\delta}{2})[c_i^\dag c_{i+1}+ c_{i+1}^\dag c_{i}]-\mu \sum_{i=1}^n c_i^\dag c_i.
\end{align}Here, we set $V_{nn}=1$ and $\delta=0.4$. In addition, with respect to system size $n$, we adjust $\mu$ such that the zero momentum state becomes the ground state of the Hamiltonian
\begin{align}
    \ket{\psi_{GS}}=\frac{1}{\sqrt{n}}\sum_{j=1}^n c_j^\dag \ket{0}.
\end{align}
For our simulations, we implement a Pauli algebra backend that leverages symbolic dictionary representations for the efficient construction and manipulation of fermionic operators. To execute the classical shadow sampling, we compute the exact statevectors and accelerate the measurement process by distributing the required shots across parallelized routines. For the brute-force measurements, we emulate the Hadamard tests with Bernoulli random variables whose expectation values correspond to given unitary observables.

\cref{fig: numerical test} illustrates the scaling of the maximum of errors between exact and estimated correlation functions with respect to system size, which is defined by
\begin{align}\label{eq: test example}
    \max_{i,j\in[n]}\abs{f_{ij}-\widetilde{f}_{ij}},\quad f_{ij}=\tr\left(\rho \{e^{iHt}c_i^\dag c_je^{-iHt},c_a^\dag c_b\}\right).
\end{align}Here $\widetilde{f}_{ij}$ represents the estimated expectation value and the indices, $a,b\in[n]$, are fixed. We set $t=1$ for the simulations in \cref{fig: numerical test}.

Regarding the shot cost, we fix total measurement budgets such that $N_{\text{FAST}}=3N_{\text{brute-force}}=80000$, which denote the total number of shots for the FAST and brute-force measurements in each of the simulation results in \cref{fig: numerical test}. With this setting, we can make the variance of estimates obtained from the two methods comparable in estimating the correlation function, $\tr\left(\rho [e^{iHt}\gamma_i\gamma_je^{-iHt},\gamma_k\gamma_l]\right)$. One reason is that the variance for brute-force measurements with Hadamard tests is no greater than $1$, while the FAST involves a sum of three variances of independent classical shadows, due to the identity \eqref{eq: bose rep}. 
With the fixed total shot numbers, we distribute them across the required measurement circuits. 
For the case of brute-force measurements, we require a measurement circuit for each combination $(i,j,k,l)$. Here the number of pairs $(k,l)$ is either $1$ or $4$, since the two indices correspond to a fixed one-body observable in  \eqref{eq: test example}. In contrast, the FAST requires only either $3$ or $12$ circuits for measurements, according to \cref{lem: ternary}. Additionally, the FAST does not require  controlled Pauli operations, while the brute-force measurements do. This is summarized in \cref{table: numerical test}. 

We perform 10 independent simulations for each system size $n$, plotting the resulting mean maximum error and its standard deviation in \cref{fig: numerical test}. 
Given the fixed shot number, the error estimate, $\max_{i,j}\abs{f_{ij}-\widetilde{f}_{ij}}$, increases with system size. Importantly, \cref{lem: ternary} implies that the  error within our protocol  scales as $\mathcal{O}\left(\sqrt{n\log n}\right)$, while the error for brute-force measurements scales as $\mathcal{O}(n\sqrt{\log n})$. We can clearly see this  quantitative divergence in \cref{fig: numerical test}, by observing that the error occurring from brute-force measurements grows much faster than our method. From this observation, we expect that for large systems, the FAST would outperform brute-force measurements.

\begin{figure}[htbp]
    
    %\centering

    \begin{center}
   \begin{subfigure}[b]{0.5\textwidth}
    \centering
    \includegraphics[width=\textwidth]{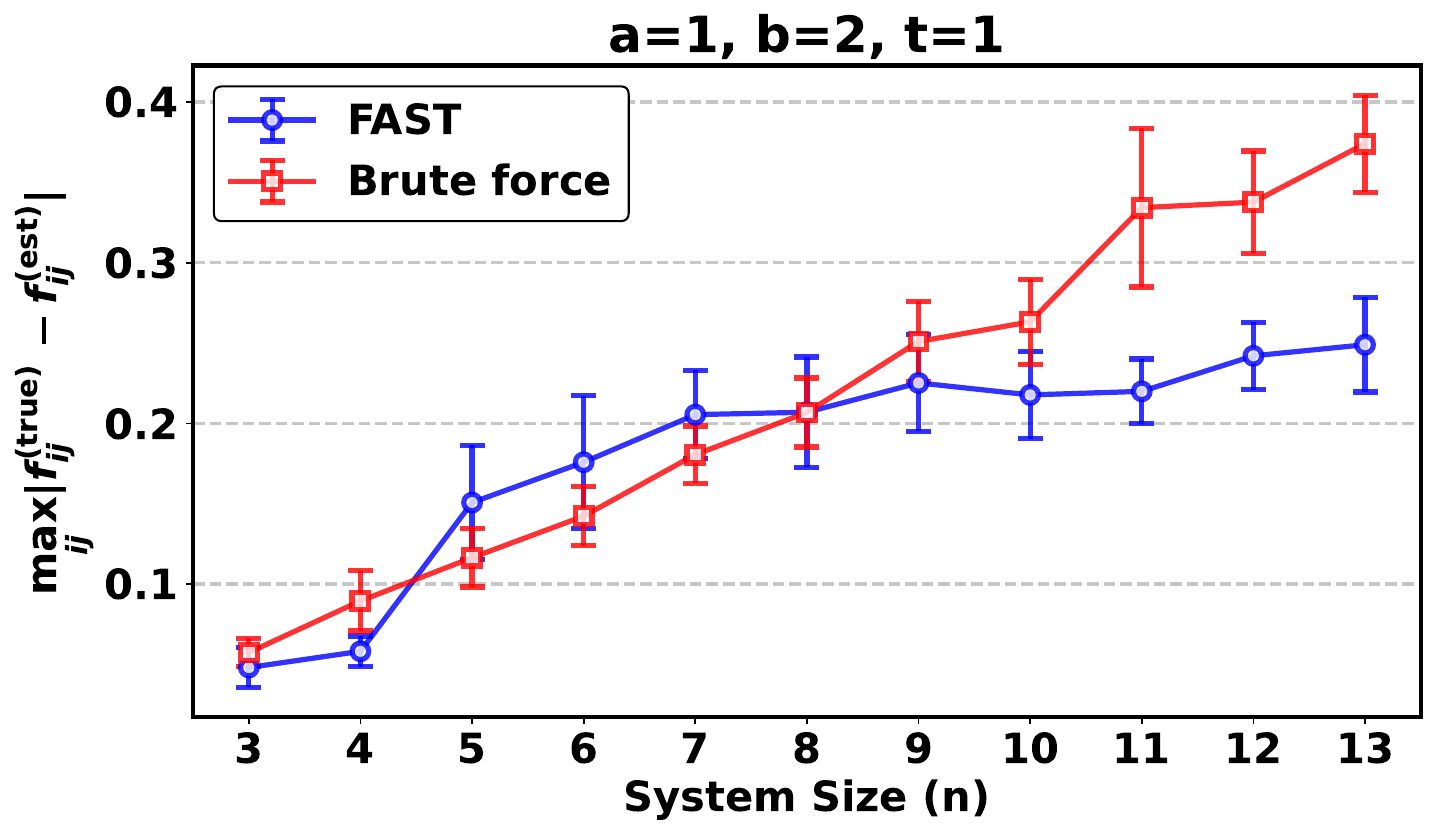}
    \end{subfigure}
    \begin{subfigure}[b]{0.5\textwidth}
    \centering
    \includegraphics[width=\textwidth]{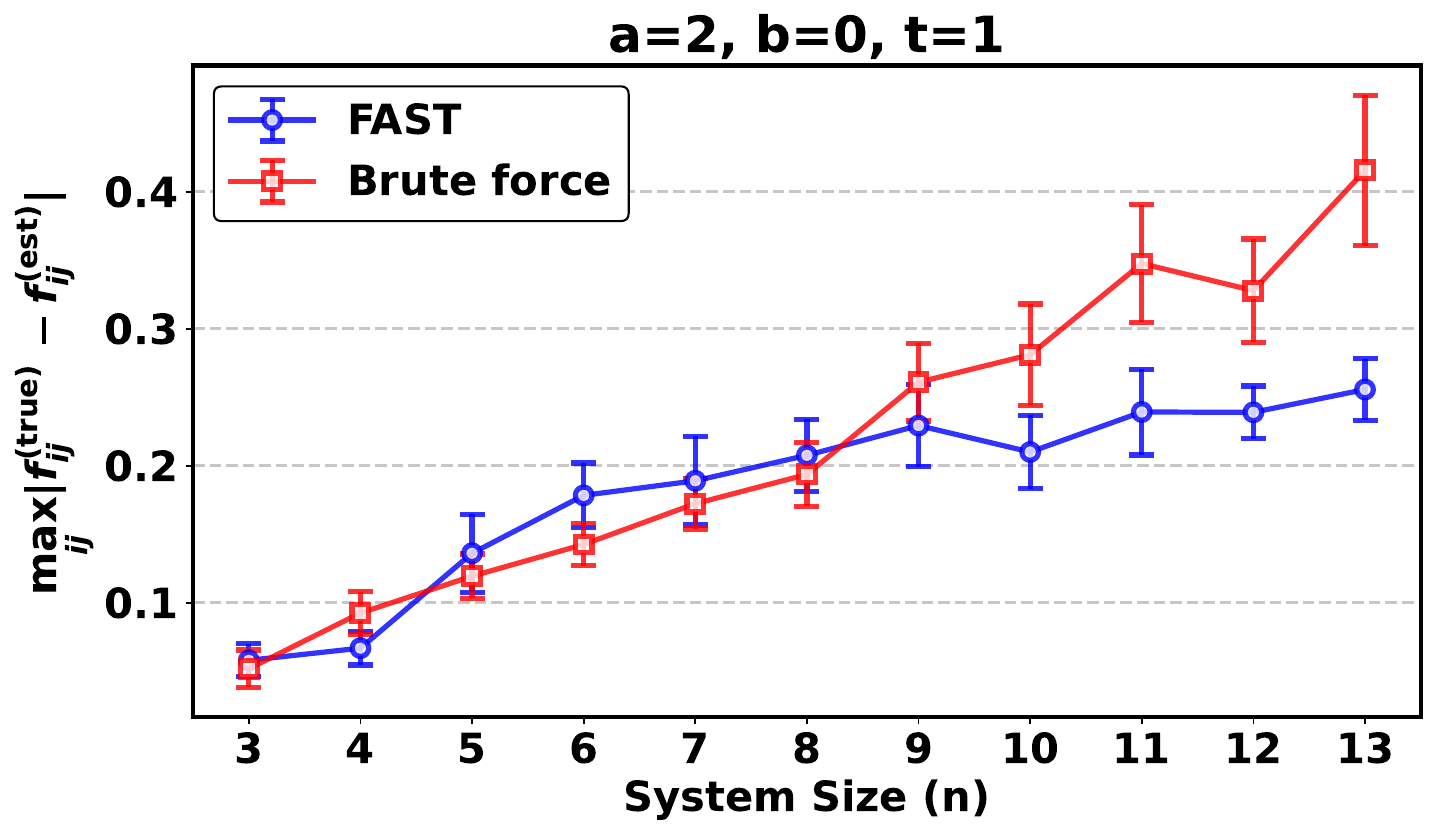}
    \end{subfigure}

\end{center}

    \caption{Simulation results that show the maximum of errors between exact and estimated correlation functions of the form in \eqref{eq: test example}, which are obtained from the FAST and brute-force measurements. The parameters $(a,b)$ are set to $(1,2)$ and $(2,2)$ for the top and bottom panels, respectively. The y-axis indicates the maximum error between exact and estimated correlation functions.} 
    \label{fig: numerical test}
\end{figure}

\begin{table}[h]
\centering
\begin{tabular}{|c|c|c|}
\hline
Method & FAST & Brute force \\
\hline
\hline
The number of measurement circuits & $3n_{ab}$ & $4n_{ab}n^2$ \\
\hline
Controlled Pauli operation & No & Yes \\
\hline
Ancilla qubit & $0$ & $1$\\
\hline
\end{tabular}
\caption{Comparison of our method (FAST) with a brute-force strategy of measurements for the simulations in \cref{fig: numerical test}. The $n_{ab}$ indicates the number of Majorana operators that form the $1$-body observable $c_a^\dag c_b$. If $a=b$, $n_{ab}=1$, otherwise $n_{ab}=4$.}
\label{table: numerical test}
\end{table}

\section*{Discussion}

In this work, we present a framework for estimating the dynamical correlation functions using shadow tomography techniques. We introduce Fermionic-Adapted Shadow Tomography (FAST) protocols (\cref{alg: algorithm1} and \cref{alg: algorithm2}), for efficiently estimating the commutator and anti-commutator correlation functions, respectively. The connection is established by reformulating these functions in a form amenable to direct measurement, thereby enabling their estimation via shadow tomography techniques. We expect the proposed protocols to be broadly applicable to practical problems in quantum simulation of many-body systems.

We remark that from the commutator representation \eqref{eq: bose rep}, and the anti-commutator ones in \eqref{eq: plus rep} and \eqref{eq: minus rep}, we can straightforwardly estimate the general form of the correlation functions \eqref{eq: function rep} by summing them up, $\tr(\rho A(t)B)$ where $A$ and $B$ correspond to $1$-body fermionic observables. In that process, we will use the measurement strategies in \cref{alg: algorithm1} only, since they target $1$-body fermionic observables. Therefore, we can extend our protocol to sample-efficient estimation for two-particle correlation functions of the form,  $\tr(\rho c_i(t)^\dag c_j(t) c_k(0)^\dag c_l(0))$, and the sample-complexity improvement would follow from \cref{alg: algorithm1}. We formally state this as in the following corollary,
% We remark that from the commutator representation \eqref{eq: bose rep}, and the anti-commutator ones in \eqref{eq: plus rep} and \eqref{eq: minus rep}, we can straightforwardly estimate the general form of the correlation functions \eqref{eq: function rep} by summing them up, $\tr(\rho A(t)B)$ where $A$ and $B$ correspond to $1$-body fermionic observables. In that process, we will use the measurement strategies in \cref{alg: algorithm1} only, since they target $1$-body fermionic observables. Therefore, we can extend our protocol to sample-efficient estimation for two-particle correlation functions of the form,  $\tr(\rho c_i(t)^\dag c_j(t) c_k(0)^\dag c_l(0))$, and the sample-complexity improvement would follow from \cref{alg: algorithm1}. We formally state this as in the following corollary,
\begin{corollary}\label{col: boson}    
    Estimating the density-density response function \eqref{eq: ddres} to additive error $\epsilon$ requires $\mathcal{O}\left(\frac{n^3\log n}{\epsilon^2}\right)$ sample complexity using \cref{alg: algorithm1} with $\mathcal{O}(n^3)$ circuits (JW, BK) or $\mathcal{O}(n^2)$ circuits (TT), while the brute force measurement demands $\mathcal{O}\left(\frac{n^4\log n}{\epsilon^2}\right)$ with $\mathcal{O}(n^4)$ circuits.
\end{corollary}

As a potential application of \cref{alg: algorithm2}, we may consider the estimation of the retarded Green’s function in momentum space. To illustrate this, for the system with one fermionic mode at each lattice site, the single-particle Green's functions in the momentum space is given by
\begin{align}\label{eq: rgf_f}
    G_{k,\sigma}^R(t) = \frac{1}{n}\sum_{a,b=1}^nG_{ab,\sigma}^R(t)e^{-ik(a-b)}.
\end{align}To estimate $G_{k,\sigma}^R(t)$ within a target tolerance $\epsilon_{\mathrm{tol}}$, each entry $G_{ab,\sigma}^R(t)$ must be computed to precision $\epsilon = \frac{\epsilon_{\mathrm{tol}}}{n}$ by the Gershgorin circle theorem in the worst-case scenario. This places the problem in the regime $n \le \frac{1}{\epsilon^2}$ of \cref{alg: algorithm2}, where our protocol does not improve sample complexity but reduce the number of measurement circuits by an order of magnitude with respect to system size. However, we anticipate no general improvement in sample complexity for estimating Green's functions. This limitation arises from the mutual anti-commutation relations of Majorana operators, a point also noted by Kökcü et al. \cite{kokcu2024linear}.

\section*{Methods}

\subsection*{The triply efficient shadow tomography for learning $1$-body observables}

In this section, we review on a specialized application of the shadow tomography protocol by King \emph{et al} \cite{king2025triply}, in order to derive the result in \cref{lem: king}. 

\begin{lemma}{\cite[Lemma 1.1]{king2025triply}}\label{lem: size of clique}
    The size of the largest clique in $G_\epsilon$ is at most $\frac{4}{\epsilon^2}$ with high probability.
\end{lemma}
\begin{lemma}{\cite[Lemma 4.1]{king2025triply}}\label{lem: coloring alg}
Let $G'$ be a subgraph of the commutation graph of $G(\mathcal{F}_1^{(n)})$ of $1$-body observables, and $\omega$ is the size of the largest clique in $G'$. Then, there is a classical algorithm that outputs a coloring of $G'$ with $\mathcal{O}(n^2\omega)$ runtime using at most $\omega+1$ colors. 
\end{lemma}
Using these lemmas, we prove \cref{lem: king}, which follows from the original work \cite{king2025triply}.
\begin{proof}
Let $G_{\epsilon}$ be the subgraph of the commutation graph for the $1$-body observables after thresholding observables of magnitudes less than $\epsilon$.

By \cref{lem: size of clique} and \cref{lem: coloring alg}, coloring $G_\epsilon$ requires $\frac{4}{\epsilon^2}+1$ colors, and for each color, there are at most $\mathcal{O}(n^2)$ observables, since the number of $1$-body observables is $\mathcal{O}(n^2)$. This means that measuring the observables for each color requires $\mathcal{O}(\frac{\log \frac{n}{\delta}}{\epsilon^2})$ sample complexity using Cliiford measurements \cite{aaronson2004improved}. Since there are $\mathcal{O}(\frac{1}{\epsilon^2})$ colors, the total sample complexity is $\mathcal{O}(\frac{\log \frac{n}{\delta}}{\epsilon^4})$. Since the number of $1$-body observables scales as $\mathcal{O}(n^2)$, by setting the new failure probability as $\delta \leftarrow \frac{\delta}{n^2}$, we can learn the observables in $G_\epsilon$ with probability at least $1-\delta$ using $\mathcal{O}(\frac{\log n}{\epsilon^4})$ samples, and $\mathcal{O}(\frac{1}{\epsilon^2})$ circuits corresponding to the colors.   
Therefore, measuring all observables in $G_\epsilon$ with single-copy measurements can be done with $\mathcal{O}(\frac{\log n}{\epsilon^4})$ sample complexity and $\mathcal{O}(\frac{1}{\epsilon^2})$ Clifford measurement circuits.

\end{proof}

\subsection*{Probabilistic state preparation based on measurements}

Within the circuit \cref{fig: circuit c}, the state is evolved as follows,
\begin{equation}\label{eq: fig1c}
\begin{split}
    &\ketbra{0}\otimes \rho\\
    &\rightarrow \ketbra{+}\otimes \rho\\
    &\rightarrow \ketbra{0}\otimes \rho + \ketbra{1}{0}\otimes  B\rho + \ket{0}\bra{1}\otimes  \rho B + \ketbra{1}\otimes B \rho B \\
    &\rightarrow \ketbra{0}\otimes e^{iHt}\rho e^{-iHt} + \ket{1}\bra{0}\otimes  e^{iHt} B\rho e^{-iHt} + \ket{0}\bra{1}\otimes  e^{iHt}\rho B e^{-iHt} + \ketbra{1}\otimes e^{iHt} B \rho B e^{-iHt}\\
    &\rightarrow \ketbra{+}\otimes e^{iHt}\rho e^{-iHt} + \ket{-}\bra{+}\otimes  e^{iHt} B\rho e^{-iHt} + \ket{+}\bra{-}\otimes  e^{iHt}\rho B e^{-iHt} + \ketbra{-}\otimes e^{iHt} B \rho B e^{-iHt}\\
    & \Longrightarrow  \rho_{+} \text{ if 0 in the ancilla with probability } C_{+}^2,\; \rho_{-} \text{ if 1 in the ancilla with probability } C_{-}^2,
\end{split}
\end{equation}where $\rho_{\pm}$ and $C_{\pm}$ are defined in \eqref{eq: plus rep} and \eqref{eq: minus rep}. Similarly, we can obtain states $\rho_{\pm}$ within the circuits  \cref{fig: circuit d}, and \cref{fig: circuit f}.

\subsection*{Chained measurement strategy for the JW-mapped Majorana operators}

In this section, we prove \cref{lem: ours}. To proceed, we first consider the following preliminary exercise.

For the independent events that $C=1$ and $D=1$, we let $\mathbb{P}(C=1)\geq 1-\delta_1$ and $\mathbb{P}(D=1)\geq 1-\delta_2$. Then, we have
\begin{equation}\label{ineq: sign}
    \mathbb{P}(CD=1)\geq \mathbb{P}(C=1,D=1)=\mathbb{P}(C=1)\mathbb{P}(D=1)\geq (1-\delta_1)(1-\delta_2).
\end{equation}This implies that if random variables, $\hat{a}$ and $\hat{b}$, satisfy that $\text{sign}(\hat{a})=\text{sign}(a)$ and $\text{sign}(\hat{b})=\text{sign}(b)$ with high probability for the corresponding true values, $a$ and $b$, then $\text{sign}(\hat{a}\hat{b})=\text{sign}(ab)$ with h.p. This can be understood by setting $C=\text{sign}(\hat{a})\text{sign}(a)$ and $D=\text{sign}(\hat{b})\text{sign}(b)$.

The RGF calculations in \eqref{eq: RGF} consider the annihilation and creation observables, which are represented by the Majorana operators under the JW mapping
\begin{equation}\label{eq: sx-sy}
 S_X:=\{X_1, Z_1X_2,...,Z_1\cdots Z_{n-1}X_n\}, \quad S_Y:=\{Y_1, Z_1Y_2,...,Z_1\cdots Z_{n-1}Y_n\}.   
\end{equation}For convenience, we denote $\{P_i\}_{i=1}^n$ and $\{Q_i\}_{i=1}^n$ to be the sets $S_X$ and $S_Y$. The indices are arranged consistently with the elements of the sets. For example, $P_1=X_1$, $P_3=Z_1Z_2X_3$, $Q_n=Z_1\cdots Z_{n-1}Y_n$, and so forth. Notice that the observables in both sets mutually anticommute. In this case, the use of the coloring strategy  \cite{king2025triply} does not give any efficiency for the problem here, since any  subset of the Majorana observables consists of mutually anticommuting Paulis.

Here we introduce a novel shadow tomography protocol based on Bell sampling. Our protocol is tailored to the set of Majorana observables in \eqref{eq: sx-sy}. Suppose that the Bell sampling for neglecting observables of small magnitudes is already performed \cite{king2025triply,huang2021information}. We let $\{P_{r_i}\}_{i=1}^I$ and $\{Q_{k_j}\}_{j=1}^J$ be the remaining observables as subsets of $S_X$ and $S_Y$, respectively, where $r_i$'s and $k_j$'s are arranged in increasing order.  We define the two sets of observables as follows,
\begin{equation}\label{eq: bx-by}
        B_X := \{P_{r_i}\otimes P_{r_{i+1}} : i\in [I-1] \},\quad B_Y := \{Q_{k_j}\otimes Q_{k_{j+1}} : j\in [J-1] \}. 
\end{equation}
In the following, we show that a property of the two sets \eqref{eq: bx-by}, 
\begin{lemma}\label{lem: bell-basis}
    The Pauli strings in $B_X$ mutually commute with the common eigenstates as follows, 
    \begin{equation}
        \{\ket{\text{Bell}}_{r_1,n+r_2}\ket{\text{Bell}}_{r_2,n+r_3}\cdots \ket{\text{Bell}}_{r_{I-1},n+r_I}\ket{b}:\; b\text{ is the computational basis on the remaining qubits} \},
    \end{equation}where $\ket{\text{Bell}}_{a,b}$ denotes the four possible Bell states acting on qubit $a$ and qubit $b$. Similarly, the Pauli strings in $B_Y$ have eigenstates as 
    \begin{equation}
        \{\ket{\text{Bell}}_{k_1,n+k_2}\ket{\text{Bell}}_{k_2,n+k_3}\cdots \ket{\text{Bell}}_{k_{J-1},n+k_J}\ket{b}:\; b\text{ is the computational basis on the remaining qubits} \}.
    \end{equation}
    
\end{lemma}
\begin{proof}
    We prove the statement by induction on the cardinality of set $B_X$. For $I=3$, it is trivial. Suppose the statement holds up to $I\leq \mathfrak{I}$. We consider the case $I=\mathfrak{I}+1$. Using the hypothesis for the subset of $B_{X}$, which consists of the first $\mathfrak{I}$ elements of $B_{X}$, we know that 
    \begin{align}
        \ket{\text{Bell}}_{r_1,n+r_2}\ket{\text{Bell}}_{r_2,n+r_3}\cdots \ket{\text{Bell}}_{r_{\mathfrak{I}-1},n+r_\mathfrak{I}}\ket{b}
    \end{align}forms the common eigenstate basis for the observables in the subset of $B_X$. Here, in the computational basis state $\ket{b}$, we can define the bell states that act on qubit $r_{\mathfrak{I}}$ and qubit $r_{n+r_{\mathfrak{I}+1}}$. That is, the common eigenstate basis mentioned above can be changed into
    \begin{align}
        \ket{\text{Bell}}_{r_1,n+r_2}\ket{\text{Bell}}_{r_2,n+r_3}\cdots \ket{\text{Bell}}_{r_{\mathfrak{I}-1},n+r_\mathfrak{I}}\ket{\text{Bell}}_{r_{\mathfrak{I}},n+r_{\mathfrak{I}+1}}\ket{b}.
    \end{align}This eigenstate is still shared by the subset of $B_X$ mentioned above. Now, it suffices to check whether this state is an eigenstate for the last observable $P_{r_{\mathfrak{I}}}\otimes P_{r_{\mathfrak{I}+1}}$. In fact, this is trivial by noticing that 
    for the bell states $\{\ket{\text{Bell}}_{r_j,n+r_{j+1}}\}_{j=1}^{\mathfrak{I}-1}$, they are transformed by the Pauli operator $Z\otimes Z$ in  $P_{r_{\mathfrak{I}}}\otimes P_{r_{\mathfrak{I}+1}}$ and  the last state $\ket{\text{Bell}}_{r_{\mathfrak{I}},n+r_{\mathfrak{I}+1}}$ by the operator $X\otimes X$ in  $P_{r_{\mathfrak{I}}}\otimes P_{r_{\mathfrak{I}+1}}$. Therefore, this completes the proof.
       
\end{proof}
Now we prove \cref{lem: ours} using \cref{lem: bell-basis}. 
\begin{proof}
For a given state $\rho$, our goal is to estimate the Majorana observables in \eqref{eq: sx-sy}, that is, $\tr(\rho O)$ for each $O\in S_X\cup S_Y$. We perform single-copy measurements and two-copy measurements as follows,
\begin{enumerate}
    \item Estimate the magnitudes of expectations of observables in \eqref{eq: sx-sy} to precision $\epsilon$

    (Two-copy measurements with $\mathcal{O}(\frac{\log n}{\epsilon^4})$ sample complexity)
    \item Throw away observables having small magnitudes ($\leq \frac{3\epsilon}{4}$). 

    ($\mathcal{O}(n)$ classical time)

    \item Let us consider the remaining observables, denoted as $\{P_{r_i}\}_{i=1}^I$ and $\{Q_{k_j}\}_{j=1}^J$ in \eqref{eq: bx-by}. 
    
    \item Measure the first observables, $P_{r_1},Q_{k_1}$, to precision $\epsilon^2$, $\tr(\rho P_{r_1})$ and $\tr(\rho Q_{k_1})$. 
    
    (Single-copy measurements with $\mathcal{O}(\frac{1}{\epsilon^2})$ sample complexity)
    \item Measure the observables in \eqref{eq: bx-by} to precision $\epsilon^2$ using the measurement basis in \cref{lem: bell-basis}.

    (Two-copy measurements with $\mathcal{O}(\frac{
    \log \min\{n,\max\{I,J\}\}
    }{\epsilon^4})$ sample complexity, where $I,J\leq n$. Note that particularly when $n\ge\frac{1}{\epsilon^2}$, according to \cite[Lemma 2.1]{king2025triply}, $\max\{I,J\}=\mathcal{O}(\frac{1}{\epsilon^2})$.  )
\end{enumerate}
In total, our protocol incurs $\mathcal{O}(\frac{\log n}{\epsilon^4})$ sample complexity  and  $\mathcal{O}(n)$ classical time cost, and $\mathcal{O}(1)$ circuits.

From the data analysis point of view, we perform three independent experiments (i.e. step 1, step 4, step 5 in the above protocol). The first experiment gives us the estimates of magnitudes, $\{p_{r_i}^{(1)}\}$, $\{q_{k_j}^{(1)}\}$ with probability at least $1-\delta_1$
\begin{equation}\label{eq: stats 1}
    \abs{\tr(\rho P_{r_i})},\abs{\tr(\rho Q_{k_j})}\geq \frac{\epsilon}{2}, \abs{p_{r_i}^{(1)}-\abs{\tr(\rho P_{r_i})}}\leq \frac{\epsilon}{4},\abs{q_{k_j}^{(1)}-\abs{\tr(\rho Q_{k_j})}}\leq \frac{\epsilon}{4}, 
\end{equation}the second experiment produces the estimates of the first observables, $p_{r_1}^{(2)}$, $q_{k_1}^{(2)}$, with probability at least $1-\delta_2$
\begin{equation}\label{eq: stats 2}
    \abs{p_{r_1}^{(2)}-\tr(\rho P_{r_1})}\leq \frac{\epsilon}{4}, \abs{q_{k_1}^{(2)}-\tr(\rho Q_{k_1})}\leq \frac{\epsilon}{4},
\end{equation}and the last gives the estimates of the chained trace multiplications, $\{p_{r_i}^{(3)}\}$, $\{q_{k_j}^{(3)}\}$, with probability at least $1-\delta_3$
\begin{equation}\label{eq: stats 3}
    \abs{p_{r_i,r_{i+1}}^{(3)}-\tr(\rho P_{r_i})\tr(\rho P_{r_{i+1}})}\leq\frac{\epsilon^2}{8}, \abs{q_{k_j,k_{j+1}}^{(3)}-\tr(\rho Q_{k_j})\tr(\rho Q_{k_{j+1}})}\leq\frac{\epsilon^2}{8}. 
\end{equation}
Since the events in \eqref{eq: stats 1}, \eqref{eq: stats 2}, \eqref{eq: stats 3} happens independently. the above inequalities hold simultaneously with probability at least $(1-\delta_1)(1-\delta_2)(1-\delta_3)$.

We claim that from the inequalities, we can estimate the Majorana observables in \eqref{eq: sx-sy}. From \eqref{eq: stats 1}, we estimate the magnitudes of expectation values of observables, and consider only the observables of magnitudes no less than $\frac{3\epsilon}{4}$. From \eqref{eq: stats 2}, we estimate the signs of the first expectation values. With \eqref{eq: stats 1} and \eqref{eq: stats 3}, we know that $\{p_{r_i}^{(3)}\}$, $\{q_{k_j}^{(3)}\}$ have the same sign as the correspondig true values. For example, \eqref{eq: stats 1} implies that 
\begin{equation}
\abs{\tr(\rho P_{r_i})\tr(\rho P_{r_{i+1}})}\geq \frac{\epsilon^2}{4},
\end{equation}and from \eqref{eq: stats 3}, we see that
\begin{equation}
    \abs{p_{r_i,r_{i+1}}^{(3)}-\tr(\rho P_{r_i})\tr(\rho P_{r_{i+1}})}\leq\frac{\epsilon^2}{8}.
\end{equation}Therefore, the signs of the estimate  $p_{r_i,r_{i+1}}^{(3)}$ and the corresponding true value $\tr(\rho P_{r_i})\tr(\rho P_{r_{i+1}})$ are the same. It is the same for the other estimates in \eqref{eq: stats 3}. Together with \eqref{eq: stats 2} and \eqref{ineq: sign}, we estimate the signs of the estimates in \eqref{eq: stats 1} by tracking the sign of $p_{r_1}^{(2)}p_{r_1,r_2}^{(3)}$, that of $p_{r_1}^{(2)}p_{r_1,r_2}^{(3)}p_{r_2,r_3}^{(3)}$, and so forth. Therefore, we prove the claim and completes the proof of \cref{lem: ours}.

\end{proof}

\section*{Acknowledgements}
This research was supported by Quantum
Simulator Development Project for Materials Innovation through the
National Research Foundation of Korea (NRF) funded by the Korean
government (Ministry of Science and ICT(MSIT))(No. NRF-
2023M3K5A1094813). 
SC was supported by a KIAS Individual Grant (CG090601) at Korea Institute for Advanced Study. TK is supported by a KIAS Individual Grant (CG096001) at Korea Institute for Advanced Study. M.H. is supported by a KIAS
Individual Grant (No. CG091301) at Korea Institute for Advanced Study. We used resources of the Center for Advanced Computation at Korea Institute for Advanced Study and the National Energy Research Scientific Computing Center (NERSC), a U.S. Department of Energy Office of Science User Facility operated under Contract No.DE-AC02-05CH11231.

\section*{Code Availability}
Code used for the current study are available at the following GitHub repository: 

https://github.com/TKmath/FAST-for-dynamical-correlation-functions.

\section*{Data Availability}
All data that support the findings of this study are included within the Github repository in Code availability.

\section*{Declarations}

\subsection*{Author Contributions}
TK conceptualized the idea and performed the theoretical and numerical analyses. TK and SC wrote the main manuscript. TK, MH, HP, and SC contributed to finalizing the manuscript.

\subsection*{Competing Interests}
The authors declare no competing financial or non-financial interests.

\bibliographystyle{unsrt}
\bibliography{ref}

\section*{Figure legends}

Figure 1: Circuits for single-copy and two-copy measurements used in our protocols for estimating the dynamical correlation functions. (a) Circuit for single-copy measurements without ancilla qubit (b) Dynamic circuit for single-copy measurements without ancilla qubit (c) Circuit for single-copy measurements with an ancilla qubit (d) Dynamic circuit for single-copy measurements with an ancilla qubit
(e) Circuit for two-copy measurements without ancilla qubit
(f) Circuit for two-copy measurements with two ancilla qubits

\end{document}